\documentclass[11pt]{amsart}

\usepackage[a4paper,margin=1.15in]{geometry}
\usepackage{amsmath,amssymb,amsthm,mathtools}
\usepackage{microtype}
\usepackage{tikz}
\usetikzlibrary{arrows.meta}
\usepackage{xcolor}
\definecolor{linkblue}{RGB}{0,0,180}
\usepackage[colorlinks=true,linkcolor=linkblue,citecolor=linkblue,urlcolor=linkblue]{hyperref}
\usepackage[capitalize,noabbrev]{cleveref}

\newtheorem{theorem}{Theorem}
\newtheorem{lemma}{Lemma}

\newtheorem{remark}{Remark}

\crefname{theorem}{Theorem}{Theorems}
\crefname{lemma}{Lemma}{Lemmas}
\crefname{proposition}{Proposition}{Propositions}
\crefname{remark}{Remark}{Remarks}
\crefname{figure}{Figure}{Figures}
\crefname{equation}{Equation}{Equations}

\newcommand{\ip}[2]{\langle #1,#2\rangle}
\newcommand{\bbE}{\mathbb E}

\title{A note on vector trifferent codes over the Sphere}
\author{Stefano Della Fiore}
\address{DII, Universit\`a degli Studi di Brescia, Via Branze 38, 25123 Brescia, Italy}
\email{stefano.dellafiore@unibs.it}
\date{}

\begin{document}

\begin{abstract}
Let $S^2\subset \mathbb R^3$ be the unit sphere.  A set
$C\subset (S^2)^n$ is called vector trifferent if for every three distinct
$x,y,z\in C$ there is a coordinate $i$ for which $x_i,y_i,z_i$ are mutually
orthogonal.  Bhandari and Khetan recently introduced this vectorial analogue
of trifferent codes and proved the upper bound
$|C|\le (\sqrt 2+o(1))(3/2)^n$.  We improve the bound to:
\[
        |C|\le (1+o(1))\left(\frac32\right)^n .
\]
The new ingredient in the proof is a local
packing inequality obtained by two-coloring, around each codeword, the circles
of vectors orthogonal to the corresponding coordinates.  This replaces a
global tensor-space bound by a centered estimate and gives exactly the missing
factor in the leading constant.
\end{abstract}

\maketitle

\section{Introduction}

Classical trifferent codes were introduced by Elias in connection with
zero-error information theory and list decoding \cite{Elias1988}.  A code
$C\subset \{0,1,2\}^n$ is trifferent if for every three distinct codewords
$x,y,z\in C$ there is a coordinate $i$ at which the three symbols
$x_i,y_i,z_i$ are all distinct.  Determining the largest possible size of such
codes is a long-standing problem.  Elias's pruning argument gives
$|C|\le 2(3/2)^n$ \cite{Elias1988}; subsequent work has produced several
improvements in the discrete setting, including computer-assisted constant
improvements for small block lengths and polynomial improvements in $n$
\cite{DellaFioreGnuttiPolak2022,Kurz2024,BhandariKhetanDiscrete2025}.
The related theory of separating systems and perfect hashing goes back at
least to Fredman and Komlos \cite{FredmanKomlos1984}.

Bhandari and Khetan recently proposed a geometric analogue, called vector
trifference~\cite{BhandariKhetanVector2025}.  In this setting the alphabet is
not the three-symbol set $\{0,1,2\}$ but the unit sphere
\[
        S^2=\{u\in \mathbb R^3:\|u\|=1\},
\]
and three symbols are separated when the corresponding vectors are mutually
orthogonal.  Thus $C\subset (S^2)^n$ is \emph{vector trifferent} if for every
three pairwise distinct $x,y,z\in C$ there exists $i\in\{1,\ldots,n\}$ such
that $x_i,y_i,z_i$ are mutually orthogonal in $\mathbb R^3$.

This continuous formulation is not a cosmetic variant of the discrete one.  A
direct reduction to ternary trifference would require a three-coloring of the
orthogonality graph on $S^2$, while this graph has chromatic number four
\cite{GodsilZaks2012}.  This obstruction is one reason why the vector problem
requires methods different from the classical pruning argument.  Bhandari and
Khetan proved the first analogue of the Elias exponential bound in the vector
setting:
\[
        |C|\le (\sqrt 2+o(1))\left(\frac32\right)^n.
\]
Their proof uses a direct-sum construction of tensor-product subspaces
associated with unordered pairs of codewords \cite{BhandariKhetanVector2025}.

The purpose of this note is to improve the leading constant.

\begin{theorem}\label{thm:main}
If $C\subset (S^2)^n$ is vector trifferent, then
\[
        |C|\le (1+o(1))\left(\frac32\right)^n
\]
as $n\to\infty$.
\end{theorem}

The proof is based on the following idea.  For $x,y\in C$, put
\[
        d(x,y)=|\{i:x_i\perp y_i\}|,
        \qquad
        A(x,y)=n-d(x,y).
\]
The tensor-space method of \cite{BhandariKhetanVector2025} gives, in this
notation, a global estimate of the form
\[
        \sum_{\{x,y\}\subset C}2^{A(x,y)}\le 3^n.
\]
Combined with the standard binomial lower bound on the same average, this
naturally leads to the constant $\sqrt 2$.  Our replacement is the stronger
centered estimate
\[
        \sum_{y\in C,\,y\ne x}2^{A(x,y)}\le 2^n
        \qquad\text{for every }x\in C.
\]
After summing over $x$, this gives
\[
        \sum_{\{x,y\}\subset C}2^{A(x,y)}\le |C|2^{n-1},
\]
which is exactly the saving needed at the level of constants.  The proof of
this centered estimate is elementary: for a fixed center $x$, each coordinate
contributes a circle $S^2\cap x_i^\perp$, and a two-coloring of this circle
turns the vector trifference condition into disjointness of binary boxes.

We shall combine this packing estimate with a positive-kernel argument.  The
kernel step is close in spirit to the moment lower bounds appearing in
\cite{BhandariKhetanVector2025}, but is formulated here through the elementary
Gram-matrix form of the Welch bound.

\section{A centered disjointness inequality}

The first ingredient is a centered packing inequality.

\begin{lemma}\label{lem:boxes}
Let $C\subset (S^2)^n$ be finite and vector trifferent.  Then, for every
$x\in C$,
\[
        \sum_{\substack{y\in C\\ y\ne x}} 2^{-d(x,y)}\le 1.
\]
\end{lemma}

\begin{proof}
Fix $x\in C$.  For each coordinate $i$, the vectors orthogonal to $x_i$ form a
circle
\[
        \Gamma_i=S^2\cap x_i^\perp .
\]
Choose an arbitrary angular parametrization of $\Gamma_i$ by
$\theta\in \mathbb R/2\pi\mathbb Z$, and color the circle with two colors by
alternating the colors on the four half-open arcs
\[
[0,\pi/2),\quad [\pi/2,\pi),\quad [\pi,3\pi/2),\quad [3\pi/2,2\pi).
\]
Thus rotating a point of $\Gamma_i$ by an angle $\pi/2$ always changes its
color; see \cref{fig:circle-coloring}.  Consequently, any two
orthogonal vectors in $\Gamma_i$ have opposite colors.

For each $y\in C\setminus\{x\}$, define a subset $B_y\subset\{0,1\}^n$ as
follows.  In coordinate $i$, impose one binary condition if $y_i\perp x_i$,
namely that the $i$-th bit is the color of $y_i$ on $\Gamma_i$; impose no
condition if $y_i\not\perp x_i$.  Hence
\[
        |B_y|=2^{n-d(x,y)},
\]
so the normalized size of $B_y$ inside $\{0,1\}^n$ is $2^{-d(x,y)}$.

We claim that the sets $B_y$, $y\in C\setminus\{x\}$, are pairwise disjoint.
Take distinct $y,z\in C\setminus\{x\}$.  Since $C$ is vector trifferent, there
is a coordinate $i$ such that $x_i,y_i,z_i$ are mutually orthogonal.  Then
$y_i,z_i\in \Gamma_i$, and they are orthogonal inside the plane $x_i^\perp$.
By the choice of the coloring, they have opposite colors.  Therefore $B_y$ and
$B_z$ impose opposite values on the $i$-th bit, and so
$B_y\cap B_z=\varnothing$.

Since the sets $B_y$, $y\in C\setminus\{x\}$, are pairwise disjoint and
all lie in $\{0,1\}^n$, we have
\[
        \left|\bigcup_{\substack{y\in C\\y\ne x}} B_y\right|
        =
        \sum_{\substack{y\in C\\y\ne x}} |B_y|
        \le
        |\{0,1\}^n|
        =2^n.
\]
Using $|B_y|=2^{n-d(x,y)}$, this becomes
\[
        \sum_{\substack{y\in C\\y\ne x}}2^{n-d(x,y)}\le 2^n.
\]
Dividing by $2^n$ gives
\[
        \sum_{\substack{y\in C\\y\ne x}}2^{-d(x,y)}\le 1,
\]
as required.
\end{proof}

\begin{figure}[t]
\centering
\begin{tikzpicture}[scale=1.05]
    \def\R{2.05}
    \foreach \a/\b/\shade in {0/90/gray!12,90/180/gray!42,180/270/gray!12,270/360/gray!42} {
        \fill[\shade] (0,0) -- (\a:\R) arc[start angle=\a,end angle=\b,radius=\R] -- cycle;
    }
    \draw[thick] (0,0) circle (\R);
    \foreach \a in {0,90,180,270} {
        \draw[thin] (0,0) -- (\a:\R);
    }
    \node at (45:1.25) {$0$};
    \node at (135:1.25) {$1$};
    \node at (225:1.25) {$0$};
    \node at (315:1.25) {$1$};

    \coordinate (p) at (35:\R);
    \coordinate (q) at (125:\R);
    \fill (p) circle (1.4pt);
    \fill (q) circle (1.4pt);
    \draw[-{Stealth[length=2.4mm]},thick] (35:0.72) arc[start angle=35,end angle=125,radius=0.72];
    \node at (80:0.98) {$\pi/2$};
    \node[anchor=south west] at (p) {$u$};
    \node[anchor=south east] at (q) {$v$};

    \node[align=center] at (0,-2.65) {orthogonal points differ by a rotation of $\pi/2$\\and therefore receive opposite colors};
\end{tikzpicture}
\caption{The two-coloring used on each circle $\Gamma_i=S^2\cap x_i^\perp$.
The four half-open arcs of length $\pi/2$ are colored alternately.  Hence a
quarter-turn sends every point to the opposite color, which is exactly the
property needed in the disjointness argument.}
\label{fig:circle-coloring}
\end{figure}

\begin{remark}\label{rem:centered-saving}
Since $A(x,y)=n-d(x,y)$, \cref{lem:boxes} is equivalent to
\[
        \sum_{\substack{y\in C\\y\ne x}}2^{A(x,y)}\le 2^n
        \qquad\text{for every }x\in C.
\]
Summing this inequality over $x\in C$ gives
\[
        \sum_{\{x,y\}\subset C}2^{A(x,y)}\le |C|2^{n-1}.
\]
This is the local substitute, in the continuous setting, for the right-degree
counting estimate available in the discrete ternary problem.  It is also the
step that improves the leading constant from $\sqrt 2$ to $1$.
\end{remark}

\section{A positive-kernel estimate}
\label{sec:positive-kernel}

For $u,v\in S^2$, define
\[
        \kappa(u,v):=\frac{1+\ip{u}{v}^2}{2}.
\]
For $x,y\in (S^2)^n$, set
\[
        K(x,y):=\prod_{i=1}^n \kappa(x_i,y_i).
\]
The next two lemmas give matching lower and upper estimates for the average of
$K$ over the code.  The lower estimate uses only a standard Gram-matrix form of
the Welch bound.

\begin{lemma}[Welch bound in average form]
\label{lem:welch}
Let $v_1,\ldots,v_M$ be unit vectors in a real inner product space of dimension
$D$.  If $I,J$ are independent uniform random elements of $[M]$, then
\[
        \bbE \left[ \ip{v_I}{v_J}^2 \right]
        \ge
        \max\left\{\frac1D,\frac1M\right\}.
\]
\end{lemma}

\begin{proof}
The diagonal terms alone give
\[
        \bbE \left[ \ip{v_I}{v_J}^2 \right]
        =
        \frac1{M^2}\sum_{i,j=1}^M \ip{v_i}{v_j}^2
        \ge
        \frac1{M^2}\sum_{i=1}^M \ip{v_i}{v_i}^2
        =
        \frac1M.
\]

For the dimension bound, let $G$ be the $M\times M$ Gram matrix
$G_{ij}=\ip{v_i}{v_j}$.  Then $G$ is positive semidefinite,
$\operatorname{rank}(G)\le D$, and $\operatorname{Tr}(G)=M$.  If
$\lambda_1,\ldots,\lambda_r$ are the nonzero eigenvalues of $G$, then
$r\le D$ and
\[
        \sum_{i=1}^r \lambda_i=M.
\]
By Cauchy-Schwarz inequality,
\[
        \sum_{i=1}^r \lambda_i^2
        \ge
        \frac{M^2}{r}
        \ge
        \frac{M^2}{D}.
\]
On the other hand,
\[
        \sum_{i=1}^r \lambda_i^2
        =
        \operatorname{Tr}(G^2)
        =
        \sum_{i,j=1}^M \ip{v_i}{v_j}^2,
\]
because the eigenvalues of $G^2$ are the squares of the eigenvalues of $G$, and
because $G$ is symmetric.  Dividing by $M^2$ gives
\[
        \bbE \left[ \ip{v_I}{v_J}^2  \right] \ge \frac1D.
\]
Together with the diagonal estimate, this proves the lemma.
\end{proof}

\begin{lemma}\label{lem:kernel-lower}
Let $C\subset (S^2)^n$ be finite, let $M=|C|$, and let $X,Y$ be independent
uniform random elements of $C$.  Then
\[
        \bbE \left[ K(X,Y) \right]
        \ge
        2^{-n}
        \sum_{s=0}^n
        \binom ns
        \max\{3^{-s},M^{-1}\}.
\]
\end{lemma}

\begin{proof}
Expanding the product defining $K$, we get
\[
\begin{aligned}
        K(X,Y)
        &=
        \prod_{i=1}^n \frac{1+\ip{X_i}{Y_i}^2}{2}  \\
        &=
        2^{-n}
        \sum_{S\subseteq[n]}
        \prod_{i\in S}\ip{X_i}{Y_i}^2 .
\end{aligned}
\]
Taking expectation,
\[
        \bbE \left[ K(X,Y) \right]
        =
        2^{-n}
        \sum_{S\subseteq[n]}
        \bbE \left[
        \prod_{i\in S}\ip{X_i}{Y_i}^2 \right] .
\]

Fix $S\subseteq[n]$.  For $x\in C$, write
\[
        x_S:=\bigotimes_{i\in S}x_i
        \in
        \bigotimes_{i\in S}\mathbb R^3,
\]
with the convention $x_\varnothing=1$.  The tensor space
$\bigotimes_{i\in S}\mathbb R^3$ has dimension $3^{|S|}$, and each $x_S$ is a
unit vector.  Moreover,
\[
        \ip{x_S}{y_S}
        =
        \prod_{i\in S}\ip{x_i}{y_i},
\]
and hence
\[
        \ip{x_S}{y_S}^2
        =
        \prod_{i\in S}\ip{x_i}{y_i}^2.
\]

Apply \cref{lem:welch} to the $M$ unit vectors
\[
        \{x_S:x\in C\}
        \subseteq
        \bigotimes_{i\in S}\mathbb R^3.
\]
They lie in dimension $3^{|S|}$.  Therefore
\[
        \bbE \left[
        \prod_{i\in S}\ip{X_i}{Y_i}^2 \right]
        =
        \bbE \left[ \ip{X_S}{Y_S}^2 \right]
        \ge
        \max\{3^{-|S|},M^{-1}\}.
\]
Substituting this estimate into the expansion of $\bbE \left[ K(X,Y) \right]$ gives
\[
        \bbE \left[ K(X,Y) \right]
        \ge
        2^{-n}
        \sum_{S\subseteq[n]}
        \max\{3^{-|S|},M^{-1}\}.
\]
Grouping subsets according to their cardinality $s=|S|$ proves the claim.
\end{proof}

\begin{lemma}\label{lem:kernel-upper}
Let $C\subset (S^2)^n$ be finite and vector trifferent, let $M=|C|$, and let
$X,Y$ be independent uniform random elements of $C$.  Then
\[
        \bbE \left[ K(X,Y) \right] \le \frac2M.
\]
\end{lemma}

\begin{proof}
If $x_i\perp y_i$, then $\kappa(x_i,y_i)=1/2$, while in every coordinate
$\kappa(x_i,y_i)\le 1$.  Therefore
\[
        K(x,y)\le 2^{-d(x,y)}.        
\]
Since $K(x,x)=1$, \cref{lem:boxes} gives, for every $x\in C$,
\[
        \sum_{y\in C}K(x,y)
        \le
        1+\sum_{\substack{y\in C\\y\ne x}}2^{-d(x,y)}
        \le 2.
\]
Summing over $x\in C$ gives
\[
        \sum_{x,y\in C}K(x,y)\le 2M.
\]
Since $X,Y$ are independent and uniform on $C$,
\[
        \bbE \left[ K(X,Y) \right]
        =
        \frac1{M^2}\sum_{x,y\in C}K(x,y)
        \le
        \frac{2M}{M^2}
        =
        \frac2M.
\]
\end{proof}

Combining \cref{lem:kernel-lower,lem:kernel-upper}, we obtain the basic finite
inequality
\[
        \frac{2}{M}
        \ge
        2^{-n}
        \sum_{s=0}^n \binom ns \max\{3^{-s},M^{-1}\}.       \tag{1}\label{eq:basic}
\]

\section{Extracting the asymptotics}
\label{sec:asymptotics}

We now extract the final bound from the energy inequality.  Set
\[
B_n:=\left(\frac32\right)^n .
\]
We shall use the following elementary binomial estimate.

\begin{lemma}
	\label{lem:binomial-extraction}
	Let \(M=M_n\) satisfy
	\[
	B_n < M \le 2B_n .
	\]
	Then
	\[
	2^{-n}\sum_{s=0}^n \binom ns
	\max\{3^{-s},M^{-1}\}
	\ge
	(1-o(1))B_n^{-1}+(1-o(1))M^{-1}.
	\]
\end{lemma}

\begin{proof}
	Let
	\[
	\rho:=\log_3\left(\frac32\right),
	\]
	so that
	\[
	\frac14<\rho<\frac12.
	\]
	Put
	\[
	t:=\lfloor \log_3 M\rfloor .
	\]
	Since \(B_n<M\le 2B_n\), we have
	\[
	\frac tn=\rho+o(1).
	\]
	
	For \(s\le t\), we use the lower bound
	\[
	\max\{3^{-s},M^{-1}\}\ge 3^{-s}.
	\]
	Therefore
	\[
	\begin{aligned}
		2^{-n}\sum_{s\le t}\binom ns3^{-s}
		&=
		B_n^{-1}\,
		\mathbb P\!\left[\operatorname{Bin}\left(n,\frac14\right)\le t\right].
	\end{aligned}
	\]
	Since \(t/n\to\rho>1/4\), the probability on the right tends to \(1\). Hence
	\[
	2^{-n}\sum_{s\le t}\binom ns3^{-s}
	=
	(1-o(1))B_n^{-1}.
	\]
	
	For \(s>t\), we have \(3^{-s}<M^{-1}\), and so
	\[
	\max\{3^{-s},M^{-1}\}=M^{-1}.
	\]
	Thus
	\[
	\begin{aligned}
		2^{-n}\sum_{s>t}\binom ns M^{-1}
		&=
		M^{-1}\,
		\mathbb P\!\left[\operatorname{Bin}\left(n,\frac12\right)>t\right].
	\end{aligned}
	\]
	Since \(t/n\to\rho<1/2\), this probability also tends to \(1\). Therefore
	\[
	2^{-n}\sum_{s>t}\binom ns M^{-1}
	=
	(1-o(1))M^{-1}.
	\]
	Adding the two estimates proves the lemma.
\end{proof}

We now finish the proof of the theorem.  Let
\[
M:=|C|.
\]
The energy inequality gives
\[
\frac2M
\ge
2^{-n}\sum_{s=0}^n\binom ns
\max\{3^{-s},M^{-1}\}.
\]
Using only the terms \(3^{-s}\), we first get
\[
\frac2M
\ge
2^{-n}\sum_{s=0}^n\binom ns3^{-s}
=
2^{-n}\left(1+\frac13\right)^n
=
B_n^{-1}.
\]
Hence
\[
M\le 2B_n.
\]

If \(M\le B_n\), there is nothing to prove.  Otherwise \(B_n<M\le 2B_n\), and
\cref{lem:binomial-extraction} applies.  Therefore
\[
\frac2M
\ge
(1-o(1))B_n^{-1}+(1-o(1))M^{-1}.
\]
Multiplying by \(M\), we obtain
\[
2
\ge
(1-o(1))\frac{M}{B_n}+1-o(1).
\]
Thus
\[
\frac{M}{B_n}\le 1+o(1),
\]
or equivalently
\[
|C|=M\le (1+o(1))\left(\frac32\right)^n .
\]

\section{On the leading constant}

It is useful to spell out why the argument reaches exactly the constant $1$.
Write
\[
        M=c\left(\frac32\right)^n.
\]
In the critical range, the right-hand side of \eqref{eq:basic} is
\[
        \left(1+\frac1c+o(1)\right)\left(\frac23\right)^n,
\]
where the first term comes from the contribution $3^{-s}$ and the second from
the contribution $M^{-1}$ above the threshold $s\approx \log_3 M$.  The
left-hand side of \eqref{eq:basic} is
\[
        \frac{2}{c}\left(\frac23\right)^n.
\]
Thus the leading-order inequality is
\[
        \frac{2}{c}\ge 1+\frac1c,
\]
which is precisely $c\le 1$.

By contrast, the direct-sum estimate of \cite{BhandariKhetanVector2025} gives
an upper bound on the corresponding pair average of size essentially
$2\cdot 3^n/M^2$ rather than the centered estimate used here.  Substituting
$M=c(3/2)^n$ then leads to
\[
        1+\frac1c\le \frac{2}{c^2}+\frac1c,
\]
and hence only $c\le \sqrt 2$.  The improvement in the present note is
therefore not an optimization of the cutoff in the binomial sum; it is the new
centered packing inequality of \cref{lem:boxes}.

The method presented here does not suggest a way to improve the leading
constant below $1$.  Such an improvement would require information beyond the
pointwise centered packing and the Welch-type lower bound used above.

\section*{Declaration of Generative AI} During the preparation of this work, the author used ChatGPT~5.6 to improve the language, readability, and clarity of the exposition. After using this tool, the author carefully reviewed and edited the manuscript as needed. The author takes full responsibility for the content of the article.

\end{document}